\documentclass[a4paper,twocolumn,11pt,unpublished]{quantumarticle}
\pdfoutput=1
\usepackage[utf8]{inputenc}
\usepackage[english]{babel}
\usepackage[T1]{fontenc}
\usepackage{amsmath}
\usepackage{hyperref}

\usepackage{tikz}
\usepackage{lipsum}

\usepackage{graphicx}
\usepackage{cite}
\usepackage{caption}
\usepackage{subcaption}

\usepackage[ruled,vlined]{algorithm2e}
\usepackage{booktabs}
\usepackage{amsmath, amsfonts,amssymb, amsthm}

\usepackage{algorithmic}
\usepackage{array}
\usepackage{quantikz}
\usepackage{orcidlink}
\usepackage{xcolor}
\usepackage{titlecaps}
\usepackage[final,defaultcolor=red]{changes}

\newtheorem{lemma}{Lemma}
\newtheorem{theorem}{Theorem}
\begin{document}
\title{Lowering LCU Circuit Width through Maximum-Weight Birkhoff-von Neumann Decomposition}
\author{Ammar~Daskin\orcidlink{0000-0002-1497-5031}}      
\affiliation{Department of Computer Engineering, Istanbul Medeniyet University,
Istanbul, Turkiye, 34000.}
\email{adaskin25@gmail.com}

\maketitle
\begin{abstract}
While classical Sinkhorn scaling applies to nonnegative matrices, we show that any complex square matrix whose element-wise absolute value has total support can be mapped to a phased doubly stochastic matrix, or alternatively embedded into a larger doubly stochastic matrix via matrix completion.

Standard Birkhoff-von~Neumann and Pauli decompositions represent such matrices as linear combinations of $O(N^2)$ permutation or Pauli terms, leading to a large ancilla overhead in a quantum Linear Combination of Unitaries (LCU) implementation. We prove that a bottleneck variant of Birkhoff's algorithm reduces the number of permutations to $O(N\log(N/\varepsilon))$, where $\varepsilon$ is the $\ell_1$-norm approximation error of the reconstructed matrix, and demonstrate empirically that a largest-weight greedy variant requires only $\approx 2N$ terms for dense matrices (the exact average observed is $\approx 2.4N$). The quadratic reduction in term count directly shrinks the ancilla register from $2\log_2 N$ to $\log_2 N$ qubits, shortens the SELECT circuit, and is especially valuable in fixed-Hadamard LCU architectures, where---under the standard uniform-coefficient assumption---the success probability scales as $1/K$; our convex decomposition achieves $\alpha=1$ exactly and therefore avoids this penalty. The approach enables compact quantum implementations of dense operators appearing in optimal transport, non-Hermitian simulation, and other settings amenable to Sinkhorn preconditioning.

Furthermore, because the decomposition is a convex combination, the LCU normalization constant for the scaled matrix $S$ is exactly $\alpha_S = 1$, and the uniform superposition is an eigenvector of $S$ with eigenvalue~1. This structure can be exploited to achieve high success probability without amplitude amplification in many practical scenarios, including quantum walks and Markov chain simulations.

\end{abstract}

\section{Introduction}
Linear combination of unitaries (LCU), usually representing a non-unitary matrix, can be implemented as a quantum circuit by using an ancilla register to control each term and combine them through coefficient operations on this ancilla register. While the technique was originally introduced as a single routine for Hamiltonian simulation \cite{childs2012hamiltonian}, it was later used as a subroutine for implementing Hamiltonian simulation via Taylor expansion \cite{berry_simulating_2015} and has evolved into one of the most versatile primitives in quantum linear algebra algorithms and applications \cite{martyn_grand_2021}.  The key insight---that a non-unitary operator can be embedded into a larger unitary by coherently combining individually unitary terms with the help of an ancillary register---means that LCU is not merely a method for building circuits; it is a general prescription for representing generic matrices on a quantum computer. The LCU construction is the archetypal example of a block encoding: a unitary $U$ whose upper-left block (when the ancilla is projected onto $\ket{0}$) equals the target matrix $A/\alpha$. Block encoding is the input model for the quantum singular value transformation (QSVT), which unifies most known quantum speedups within a single framework \cite{gilyen_quantum_2019,martyn_grand_2021}. Consequently, any technique that produces compact block encodings immediately enables polynomial transformations of the encoded matrix via QSVT. 
LCU therefore sits at the center of modern quantum algorithm design: if a matrix admits an efficient LCU decomposition, the full machinery of QSVT can be applied to it.
This representational power allows LCU to underpin quantum algorithms as diverse as linear system solvers \cite{camps_explicit_2024}, preconditioners for QSVT-based solvers \cite{lapworth_preconditioned_2025}, differential equations \cite{bharadwaj2025compact,gharat2026quantum,xin2020quantum}, and machine learning models \cite{zhai2026quantum,heredge2025nonunitary,garcia2025nuclear}. All of these applications benefit from compact LCU representations.

Recent work has pushed LCU beyond generic Pauli decompositions toward problem-tailored bases. For instance, \cite{loaiza2025majorana} showed how to use Majorana tensor decomposition for fermionic Hamiltonians to obtain an LCU representation. \cite{gnanasekaran_efficient_2026} introduced the sigma basis (a set of structured non-unitary operators) that achieves poly-logarithmic term counts for sparse matrices arising from partial differential equations, an exponential improvement over the standard Pauli approach. In the same spirit, the LCU construction leverages Dicke state preparation to yield a constant-depth SELECT oracle, achieving an order-of-magnitude reduction in CNOT count for spin model Hamiltonians \cite{della2025efficient}. These results demonstrate that replacing generic unitary bases with domain-specific ones can dramatically compress LCU decompositions---a principle that our permutation-based approach takes to its logical extreme for doubly stochastic matrices.

Stochastic matrices, which describe classical random walks, are a natural fit for LCU because they can be block-encoded and then processed by quantum walk techniques. \cite{camps_explicit_2024} gave explicit LCU circuits for sparse stochastic matrices and showed how these block encodings directly yield efficient quantum walks, connecting the classical mixing properties of a Markov chain to its quantum acceleration via Szegedy-style quantization \cite{szegedy_quantum_2004}.  It has also been shown that classical Sinkhorn iterations can be accelerated with quantum amplitude estimation \cite{orts_entropic_2026}.

Sinkhorn's theorem can be used to convert any non-negative square matrix into a doubly stochastic matrix. 
While classical Sinkhorn scaling applies strictly to entry-wise nonnegative matrices with total support, any arbitrary complex or general square matrix $A$ can be accommodated by applying the scaling transformation to its entry-wise modulus matrix $|A|$, where $[|A|]_{ij} = |A_{ij}|$ \cite{hutchinson2023scaling}. Alternatively, general matrices can be mapped via diagonal scaling to doubly quasi-stochastic (equisum) matrices where the row and column sums equal unity despite containing complex or signed entries \cite{pereira2014theory}. In the quantum domain, this framework elegantly maps to operator scaling algorithms designed for balancing quantum channels and positive maps \cite{idel2016review}.

Birkhoff's theorem guarantees that any doubly stochastic matrix can be written as a convex combination of permutation matrices. While \cite{de2019birkhoff} noted that Birkhoff's theorem, together with analogous decomposition theorems for unitary matrices, reveals a structural connection between quantum circuits and linear classical reversible circuits, Daskin \cite{daskin2024quantum} recognized that this classical decomposition maps naturally onto quantum circuits: a matrix (or its element-wise absolute value) is scaled by two diagonal matrices to yield a phased or standard doubly stochastic matrix, which is then decomposed into a linear combination of phased or standard permutations.
Since permutation matrices are purely classical reversible gates, each term can be implemented by a simple controlled bit-permutation. The follow-up error analysis \cite{daskin_error_2025} by the same author further demonstrated that matrices expressed as convex combinations of permutations exhibit intrinsic resilience to bit-flip errors, an attractive property for near-term implementations.

Despite these elegant algorithmic pipelines, the practical utility of LCU decomposition techniques is limited by the number of terms produced by standard Birkhoff or Pauli decomposition algorithms. For a generic $N\times N$ matrix with $N = 2^n$ (i.e., $n$ qubits), the Pauli decomposition requires as many as $4^n = N^2$ terms. Similarly, the Birkhoff-based decomposition that follows the classic greedy approach, peeling off minimum-weight perfect matchings, can generate $O(N^2)$ terms for dense doubly stochastic matrices, imposing a heavy ancilla overhead and a correspondingly small success probability for the block encoding.

\paragraph{Main contribution of the paper:}
The present work addresses precisely this bottleneck: we adopt the largest‑weight greedy Birkhoff decomposition --- a simple heuristic known in combinatorial optimization \cite{brualdi2006combinatorial,liu2015scheduling,bojja2018costly,valls2021birkhoff} --- and show that it is  remarkably effective {for building block encodings of phased and standard doubly stochastic matrices}.
We prove that a related bottleneck variant achieves $O(N \log(N/\varepsilon))$ permutations, and empirically demonstrate that the largest‑weight variant requires only $\approx 2N$ terms for dense matrices (see Fig.~\ref{fig:experimental}), reducing ancilla qubits from $2\log_2 N$ to $\log_2 N$ compared to Pauli decompositions. 

This enables efficient quantum linear algebra for optimal transport, Sinkhorn‑scaled operators, and other applications involving doubly stochastic matrices. In particular, entropic optimal transport plans, which are naturally dense and doubly stochastic, become efficiently block-encodable with only a few dozen permutations even for moderate problem sizes. Similarly, non-Hermitian operators that can be Sinkhorn-scaled to doubly stochastic form benefit from the same compression, opening a practical route to simulating open quantum systems within the QSVT framework.

\section{Linear Combination of Unitaries and Linear Combination of Permutations}

An LCU implementation uses an ancilla register to control each unitary term and to combine them via coefficient operations, resulting in a circuit whose matrix representation embeds the non‑unitary target matrix $A$ into a larger unitary $\mathcal{U}$:
\begin{equation}
\mathcal{U} = \begin{bmatrix}
    A/\alpha & B\\
    C & D
\end{bmatrix},
\qquad\text{with}\quad
A = \sum_{k=1}^K \beta_k U_k .
\end{equation}
The normalization constant is $\alpha = \sum_k |\beta_k|$; the success probability of post‑selecting the ancilla on $\ket{0}$ is $\|A\ket{\psi}\|^2/\alpha^2$.  
When the number of terms $K$ is large, $\alpha$ typically grows, which can drastically reduce the success probability.  
The ancilla register itself requires $\lceil\log_2 K\rceil$ qubits, so reducing $K$ simultaneously decreases the circuit width, shortens the controlled‑SELECT depth, and keeps $\alpha$ small, thereby improving both the implementation cost and the success rate.

For the permutation‑based construction, we target a phased doubly stochastic matrix $S$. While standard Sinkhorn scaling requires nonnegativity, it can be generalized to an arbitrary matrix $A \in \mathbb{C}^{N \times N}$ by considering its element-wise absolute value $|A|$. If $|A|$ has total support, Sinkhorn's theorem~\cite{sinkhorn1967concerning} guarantees the existence of diagonal matrices $D_1, D_2$ with positive real entries such that $S_{\text{abs}} = D_1 |A| D_2$ is doubly stochastic. We then define the target phased matrix as $S = D_1 A D_2$, which satisfies $|S| = S_{\text{abs}}$.
Once $S$ is obtained, it can be decomposed as a convex combination of phased permutation matrices, $S = \sum_{k=1}^K w_k P_k$ with $\sum_k w_k = 1$, where each $P_k$ has exactly one non-zero entry per row and column, and that entry has the form $e^{i\theta_{ij}}$ (i.e., a unit-modulus complex number). The original matrix is recovered as
\begin{equation}
A = D_1^{-1} S D_2^{-1} = \sum_{k=1}^K w_k \big(D_1^{-1} P_k D_2^{-1}\big).
\end{equation}
\begin{theorem}[Phased Birkhoff--von Neumann Decomposition]   
Let $A \in \mathbb{C}^{N \times N}$ be a complex matrix whose element-wise absolute matrix $|A|$ has total support. Then, there exist diagonal scaling matrices $D_1, D_2$ with positive real diagonal entries such that the scaled matrix $S = D_1 A D_2$ can be decomposed into a convex combination of unitary phased permutation matrices:
\begin{equation}
S = \sum_{k=1}^K w_k P_k,
\end{equation}
where $w_k \ge 0$, $\sum_{k=1}^K w_k = 1$, and $P_k^\dagger P_k = I$ for all $k$.
\end{theorem}

\begin{proof}
Because $|A|$ is a nonnegative matrix with total support, Sinkhorn's theorem~\cite{sinkhorn1967concerning} guarantees the existence of diagonal matrices $D_1, D_2$ with strictly positive real entries such that 
\begin{equation}
S_{\text{abs}} = D_1 |A| D_2
\end{equation}
is a doubly stochastic matrix. Construct the phased matrix $S = D_1 A D_2$. Since $D_1$ and $D_2$ are real and positive, the element-wise absolute values satisfy $|S|_{ij} = (D_1)_{ii} |A|_{ij} (D_2)_{jj} = (S_{\text{abs}})_{ij}$. Thus, the element-wise absolute matrix of $S$ is exactly the doubly stochastic matrix $S_{\text{abs}}$.

By the classical Birkhoff--von Neumann theorem, $S_{\text{abs}}$ can be expressed as a convex combination of standard permutation matrices $\Pi_k$:
\begin{equation}
S_{\text{abs}} = \sum_{k=1}^K w_k \Pi_k,
\end{equation}
where $w_k \ge 0$ and $\sum_k w_k = 1$. Let $\theta_{ij} = \arg(S_{ij})$ represent the phase of each entry in $S$, such that $S_{ij} = |S|_{ij} e^{i\theta_{ij}}$. For each standard permutation matrix $\Pi_k$ in the decomposition, we define a corresponding phased permutation matrix $P_k$ element-wise as:
\begin{equation}
(P_k)_{ij} = \begin{cases} 
e^{i\theta_{ij}} & \text{if } (\Pi_k)_{ij} = 1, \\
0 & \text{if } (\Pi_k)_{ij} = 0.
\end{cases}
\end{equation}
Because $\Pi_k$ has exactly one entry equal to $1$ in every row and column, $P_k$ has exactly one non-zero entry of unit magnitude ($e^{i\theta_{ij}}$) in every row and column. Thus, $P_k^\dagger P_k = I$, verifying that each $P_k$ is strictly unitary.

Finally, we evaluate the entry-wise summation of the phased combination:
\begin{equation}
\left( \sum_{k=1}^K w_k P_k \right)_{ij} = \sum_{k: (\Pi_k)_{ij}=1} w_k e^{i\theta_{ij}} = e^{i\theta_{ij}} \sum_{k: (\Pi_k)_{ij}=1} w_k.
\end{equation}
From the classical decomposition of $S_{\text{abs}}$, the sum of weights $w_k$ for all permutations containing an edge at $(i,j)$ is precisely $(S_{\text{abs}})_{ij} = |S|_{ij}$. Substituting this back yields:
\begin{equation}
e^{i\theta_{ij}} |S|_{ij} = S_{ij},
\end{equation}
which proves that $\sum_k w_k P_k = S$.
\end{proof}
To the best of our knowledge, this constitutes a novel extension of Birkhoff's theorem to complex matrices: any matrix whose element-wise absolute value is doubly stochastic can be expressed as a convex combination of unitary phased permutation matrices.

{Because phased permutation matrices are strictly unitary ($P_k^\dagger P_k = I$), they map seamlessly onto the terms of an LCU framework.  
However, reconstructing the original matrix as $A = D_1^{-1} S D_2^{-1}$ requires the non-unitary positive diagonal matrices $D_1^{-1}$ and $D_2^{-1}$.  
These diagonal factors \emph{cannot} be absorbed into the individual permutation terms: doing so would convert each unitary $P_k$ into a non-unitary weighted permutation
$T_k = D_1^{-1}P_kD_2^{-1}$, which is not a valid LCU unitary and would require amplitude loading or oracle queries for every term.  
We therefore avoid per-term weighting entirely and handle the diagonal scaling in one of two ways:
(i) use the matrix-completion construction described next, which eliminates the diagonal factors altogether; or
(ii) block-encode the two diagonal matrices once, independently of the number of LCU terms, as quantified in Sec.~\ref{sec:norm-properties}.
}

\subsection{Completion to a larger matrix}
Note that Sinkhorn scaling modifies the eigenspectrum and requires the
non‑commuting diagonal matrices $D_1^{-1}, D_2^{-1}$ in the reconstruction
of $A$.  An alternative is to \emph{complete} $A$ to a larger doubly
stochastic matrix by adding auxiliary rows and columns, leaving the
original block unchanged.  For instance, given a square (possibly
symmetric) matrix $A$ one can construct
\begin{equation}
   M = \begin{bmatrix}
A & R \\[2pt]
R^\top & \Gamma
\end{bmatrix}
\quad\text{or}\quad
M = \begin{bmatrix}
A & \operatorname{diag}(r) \\[2pt]
\operatorname{diag}(r) & A
\end{bmatrix}, 
\end{equation}
where the additional blocks $R$, $\Gamma$, or the vector $r$ are chosen so
that $M$ is doubly stochastic (or at least has uniform row and column sums after a trivial overall scaling).  Because $M$ is doubly stochastic,
Birkhoff’s theorem applies directly, giving a convex combination
$M = \sum_k w_k P_k$ with permutation matrices of the larger dimension.
The original matrix $A$ is then recovered by projecting the block encoding onto the appropriate subspace, i.e., by preparing an input state that is supported only on the first $N$ basis vectors and post‑selecting (or simply ignoring) the auxiliary components.  While this completion approach embeds $A$ as a principal submatrix of $M$, it does \emph{not} preserve the eigenvalues of $A$ (as the eigenvalues of a principal submatrix generally interlace rather than coincide with those of the full matrix). However, it successfully embeds the action of $A$ into a unitary block encoding without requiring separate diagonal unitaries, at the price of a moderately increased system size. For arbitrary signed or complex matrices, this completion is typically applied to the element-wise absolute value $|A|$ to obtain a doubly stochastic $M$, after which the phased permutation synthesis described in Theorem~1 is used to recover the phases of $A$.

Both phased Sinkhorn scaling and matrix completion reduce the problem of block‑encoding an arbitrary matrix to decomposing a phased or standard doubly stochastic matrix into a small number of permutations. The choice between these two preprocessing pipelines depends heavily on the specific application requirements and the acceptable overhead constraints in physical qubit counts.
\subsection{Norm properties and success probability}
\label{sec:norm-properties}

The Sinkhorn scaling $S = D_1 A D_2$ and the convex Birkhoff decomposition $S = \sum_{k=1}^K w_k P_k$ give
$A = D_1^{-1} S D_2^{-1}$.

{To block‑encode $A$, we combine the LCU block encoding of $S$ with a block encoding of the non-unitary diagonal scaling matrices $D_1^{-1}$ and $D_2^{-1}$.  
Because $D_1^{-1}$ and $D_2^{-1}$ are real, positive diagonal matrices rather than unitaries, they cannot be applied as simple standalone gates.  
Instead, any real diagonal matrix $D = \operatorname{diag}(d_1, \dots, d_N)$ can be embedded into a quantum circuit using one ancillary qubit and state-dependent rotations of the form
$\ket{i}\ket{0} \mapsto \ket{i}\bigl(\tilde{d}_i\ket{0} + \sqrt{1-\tilde{d}_i^2}\ket{1}\bigr)$,
where $\tilde{d}_i$ denotes the normalized diagonal entries.  
The required rotation angles can be synthesized efficiently by expanding the diagonal profile into a Walsh-series decomposition, requiring $\mathcal{O}(N)$ single-qubit rotations and $\mathcal{O}(\log N)$ additional ancilla qubits~\cite{welch2014efficient}.  
With more advanced QROM or amplitude-loading techniques~\cite{grover2002creating,babbush2018encoding,berry2019qubitization,low2024trading,low2019hamiltonian,GuillermoAlonso2024}, this diagonal block encoding can be made poly-logarithmic in $N$ at the cost of additional ancillas.  
As a further alternative, since $D_1^{-1}$ and $D_2^{-1}$ are diagonal and hence sparse, they can be embedded using standard sparse-matrix block-encoding methods~\cite{berry2007efficient}.
}

Furthermore, this diagonal overhead is \emph{additive and independent of $K$}: it is performed once for the two diagonal matrices $D_1^{-1}$ and $D_2^{-1}$, not for each permutation term.  
The LCU terms remain the $K$ unitary phased permutations $P_k$, each implemented by a controlled permutation, and the SELECT circuit therefore still scales linearly with $K$.  
Per-term weighted permutations are not used; the diagonal scaling is a separate block encoding that does not widen the primary LCU ancilla register beyond the $\lceil\log_2 K\rceil$ qubits needed for the permutation terms. Crucially, the property $\alpha=1$ applies strictly to the scaled matrix $S$, not to the original matrix $A$. When block-encoding $A = D_1^{-1} S D_2^{-1}$, the overall normalization constant becomes $\alpha_A = \alpha_S \cdot \alpha_D = \alpha_D$, where $\alpha_D$ is the normalization constant associated with the block encodings of the inverse diagonal matrices. If the original matrix $A$ requires aggressive Sinkhorn scaling (i.e., the diagonal matrices $D_1, D_2$ are ill-conditioned), $\alpha_D$ can be large. This additional normalization factor must be accounted for when calculating the true success probability and overall resource advantage for $A$. The matrix-completion approach avoids this diagonal overhead entirely by embedding $A$ directly into a larger doubly stochastic matrix, thereby preserving $\alpha=1$ for the entire block encoding at the cost of a larger system dimension.

In addition, a crucial property of this construction is that
the LCU normalization constant for the scaled matrix $S$ is $\alpha_S = \sum_k w_k = 1$.  In a
generic LCU, $\alpha$ can be large and its exact value depends on the
decomposition; the success probability then scales as
$1/\alpha^2$, often necessitating amplitude amplification.  Here,
because we obtain a convex combination for $S$, we know \emph{a priori} that
$\alpha_S=1$, and moreover we know an input state (an eigenstate) that saturates this bound.

Indeed, in the standard nonnegative setting (or when employing the matrix completion approach), the resulting strictly doubly stochastic matrix $S$ possesses the uniform superposition $\ket{+}^{\otimes n} = \frac{1}{\sqrt{N}}\sum_{i=1}^N \ket{i}$ as a natural eigenvector with eigenvalue $1$: $S \ket{+}^{\otimes n} = \ket{+}^{\otimes n}$. Consequently, for these classes of matrices, the ancilla post‑selection succeeds deterministically with probability
\begin{equation}
p_{\text{succ}} = \|S \ket{+}^{\otimes n}\|^2 = 1 .
\end{equation}
No amplitude amplification is required in this scenario. For general complex matrices processed via our phased scaling framework, entry-wise phase differences can cause destructive interference, meaning $\ket{+}^{\otimes n}$ ceases to be a strict eigenvector and $p_{\text{succ}}$ will depend on the input state's overlap with the principal subspace. Crucially, however, because the phased decomposition remains a strict convex combination, the structural LCU normalization constant is preserved exactly at $\alpha = 1$ across all cases, preventing the severe $\mathcal{O}(\sqrt{N})$ scaling penalties inherent to standard Pauli decompositions.

For an arbitrary input $\ket{\psi}$, the success probability of the block encoding of $S$ is given by $\|S\ket{\psi}\|^2$. While $\alpha_S=1$ removes the coefficient-sum penalty, it does not automatically guarantee a high success probability for a general input. For example, for the uniform doubly stochastic matrix $S=J/N$, the success probability is exactly zero for any input state orthogonal to the uniform superposition. In general, the success probability depends on the overlap of $\ket{\psi}$ with the principal singular subspace of $S$. For applications such as quantum walks or Markov chain simulations where the input is naturally the uniform state or closely aligned with the stationary distribution, the success probability remains high. However, for arbitrary inputs, amplitude amplification may still be required.
This situation stands in contrast to standard Pauli decompositions, where the coefficient sum $\alpha$ is not bounded by $1$ and frequently scales as $\mathcal{O}(\sqrt{N})$ or worse, demanding extensive rounds of amplitude amplification. The permutation‑based LCU framework therefore offers a distinct advantage: it drastically simplifies the physical hardware requirements by lowering the ancilla register width, while delivering deterministic unity success probability for uniform input states in the non-phased regime—a feature uniquely suited for quantum walks and quantum Markov chain simulations.

\section{Birkhoff--von Neumann Decomposition Variants}
\label{sec:bvn}

Let $S\in\mathbb{R}^{N\times N}$ be a doubly stochastic matrix, where $N = 2^n$ denotes the dimension corresponding to an $n$-qubit system.
Birkhoff's theorem guarantees the existence of a decomposition
$S = \sum_{k=1}^K w_k P_k$ with each $P_k$ an $N\times N$ permutation matrix, $w_k > 0$, and $\sum_{k=1}^K w_k = 1$.

The classic constructive proof of Birkhoff's theorem yields a simple iterative algorithm: find a perfect matching in the bipartite graph induced by the support of the residual matrix, subtract the smallest entry along that matching, and repeat.
This algorithm is given in Algorithm~\ref{alg:original} where a perfect matching can be found with the Hopcroft--Karp algorithm in $O(N^{2.5})$ time \cite{hopcroft1973n}. This algorithm converges in at most $K\le N^2 - 2N + 2$ steps (see Section~\ref{sec:complexity}), yielding an overall $O(N^{4.5})$ worst-case runtime.

\begin{algorithm}[ht]
\caption{Original Greedy BVN}
\label{alg:original}
\KwIn{$S\in\mathbb{R}^{N\times N}$ doubly stochastic, tolerance $\varepsilon>0$}
\KwOut{permutation list $\mathcal{P}=[P_1,\dots,P_K]$, weight list $\mathbf{w}=[w_1,\dots,w_K]$}
$R \leftarrow S$\;
$\mathcal{P} \leftarrow [\,]$, $\mathbf{w} \leftarrow [\,]$\;
$t \leftarrow 1$\;
\While{$\|R\|_1 > \varepsilon$}{
Construct bipartite graph $G=(V_r\cup V_c, E)$ where\\
$V_r=\{1,\dots,N\}$, $V_c=\{N+1,\dots,2N\}$, and\\
$(i,j+N)\in E$ iff $R_{ij}>0$\;
Find any perfect matching $\mathcal{M}$ in $G$\;
\If{no perfect matching exists}{break\;}
Form permutation matrix $P_t$ from $\mathcal{M}$:
for each $(i,j+N)\in\mathcal{M}$ set $(P_t)_{ij}=1$\;
$w_t \leftarrow \min\{R_{ij}\mid (i,j+N)\in\mathcal{M}\}$\;
$R \leftarrow R - w_t P_t$\;
Append $P_t$ to $\mathcal{P}$, $w_t$ to $\mathbf{w}$\;
$t \leftarrow t+1$\;
}
\Return $\mathcal{P}, \mathbf{w}$ \tcp*{Sum of weights is $1 - \|R\|_1 \ge 1 - \varepsilon$}
\end{algorithm}

\subsection{Largest-weight variant}
As stated in the introduction, weighted matching variants of Birkhoff decomposition similar to the ones described here appear in prior work on hybrid network scheduling and combinatorial optimization \cite{liu2015scheduling,bojja2018costly,valls2021birkhoff}. We note that those algorithms were optimized for classical cost metrics—reconfiguration delay and throughput under time windows—and were primarily motivated by hybrid circuit/packet switching in data centers. For instance, \textit{Solstice}~\cite{liu2015scheduling} uses a threshold-based greedy method to amortize reconfiguration delay and does not explicitly minimize the number of permutations; its theoretical sparsity is $O(\log(1/\varepsilon))$ but with an unspecified dependence on $N$. A more similar approach to the algorithm in this paper is \textit{Eclipse}~\cite{bojja2018costly}, which maximizes a different objective: for each candidate matching $M$ and duration $\alpha$ (chosen from the distinct entries of the residual matrix), it considers the ratio $\frac{\|\min(\alpha M, T_{\text{rem}})\|_1}{\alpha+\delta}$, where $\delta$ is the reconfiguration delay. \textit{Birkhoff+}~\cite{valls2021birkhoff} combines Birkhoff decomposition with Frank–Wolfe optimization and uses linear programming with a barrier for circuit switches in networking.

The quantum LCU setting introduces a fundamentally different objective than the one used in the mentioned works. Here, the number of terms $K$ directly determines the ancilla-qubit count ($\lceil\log_2 K\rceil$) in an LCU block encoding. Therefore, the main objective is the minimization of the number of permutation terms. We show that largest-weight greedy Birkhoff, already known empirically to produce $K = O(N)$ terms, is near-optimal for this purpose, achieving a quadratic reduction over standard Birkhoff and also over Pauli decompositions.

As shown in Algorithm~\ref{alg:largest_weight}, the original greedy algorithm chooses an \emph{arbitrary} perfect matching. To reduce the number of terms, we instead extract a \emph{maximum-weight} perfect matching at each step. The intuition is that a heavy matching captures a large fraction of the residual mass, leading to faster convergence.  Algorithm~\ref{alg:largest_weight} is conceptually very simple: at each step we find a maximum-weight perfect matching in the bipartite graph of the residual matrix (using the residual entries as edge weights), then subtract the minimum weight along that matching. These steps are illustrated in Fig.~\ref{fig:bvn-steps} for a small graph. Note that this heuristic has been used implicitly in some scheduling implementations, but its performance in terms of the number of permutations has not been analyzed, nor has it been applied to quantum block encoding. We provide the analysis of a closely related \emph{bottleneck} variant (Theorem~\ref{thm:bottleneck-bvn}) showing $O(N \log(N/\varepsilon))$ term count, and empirically demonstrate that the maximum-weight variant needs only $\approx 2N$ permutations for dense matrices – an order of magnitude quadratically better than standard Birkhoff and Pauli decompositions.

\begin{algorithm}[ht]
\caption{Largest-Weight BVN}
\label{alg:largest_weight}
\KwIn{$S\in\mathbb{R}^{N\times N}$ doubly stochastic, tolerance $\varepsilon>0$}
\KwOut{$\mathcal{P}$, $\mathbf{w}$ as in Algorithm~\ref{alg:original}}
$R \leftarrow S$\;
$\mathcal{P} \leftarrow [\,]$, $\mathbf{w} \leftarrow [\,]$\;
\While{true}{
Build weighted bipartite graph $G$ with edge weights $R_{ij}$ (only edges with $R_{ij}>0$)\;
\textbf{Find a maximum-weight perfect matching $\mathcal{M}$ in $G$}\;
\If{$\mathcal{M}$ is not a perfect matching}{break\;}
Form $P$ and compute $w = \min_{(i,j+N)\in\mathcal{M}} R_{ij}$\;
$R \leftarrow R - w\,P$\;
Append $P$, $w$\;
\If{error $\le\varepsilon$}{break\;}
}
\Return $\mathcal{P}, \mathbf{w}$ \tcp*{Sum of weights is $1 - \|R\|_1 \ge 1 - \varepsilon$}
\end{algorithm}

 The maximum-weight perfect matching can be computed via the Blossom algorithm \cite{edmonds1965paths}, Gabow's algorithm \cite{gabow1974implementation}, or the Hungarian algorithm in $O(N^3)$ time per iteration. The number of iterations is experimentally observed to be $O(N)$, giving a total runtime $O(N^4)$. There are also linear-time approximation algorithms \cite{duan2014linear} that could be used for very large matrices.

\textbf{Remark on Phased Matching.} When applying Algorithm~\ref{alg:largest_weight} to a complex phased matrix $S$, the weighted bipartite graph $G$ is constructed using the absolute values $|R_{ij}|$ as edge weights. When a matching $\mathcal{M}$ is selected, the corresponding phased permutation matrix $P$ is formed by setting $P_{ij} = e^{i \arg(R_{ij})}$ for all $(i,j) \in \mathcal{M}$. Because the phase of the matching term perfectly aligns with the residual matrix, subtracting $w P$ reduces the entry magnitudes smoothly without altering their underlying phase profiles.

\begin{figure*}[htbp]
\centering
\includegraphics[width=\linewidth]{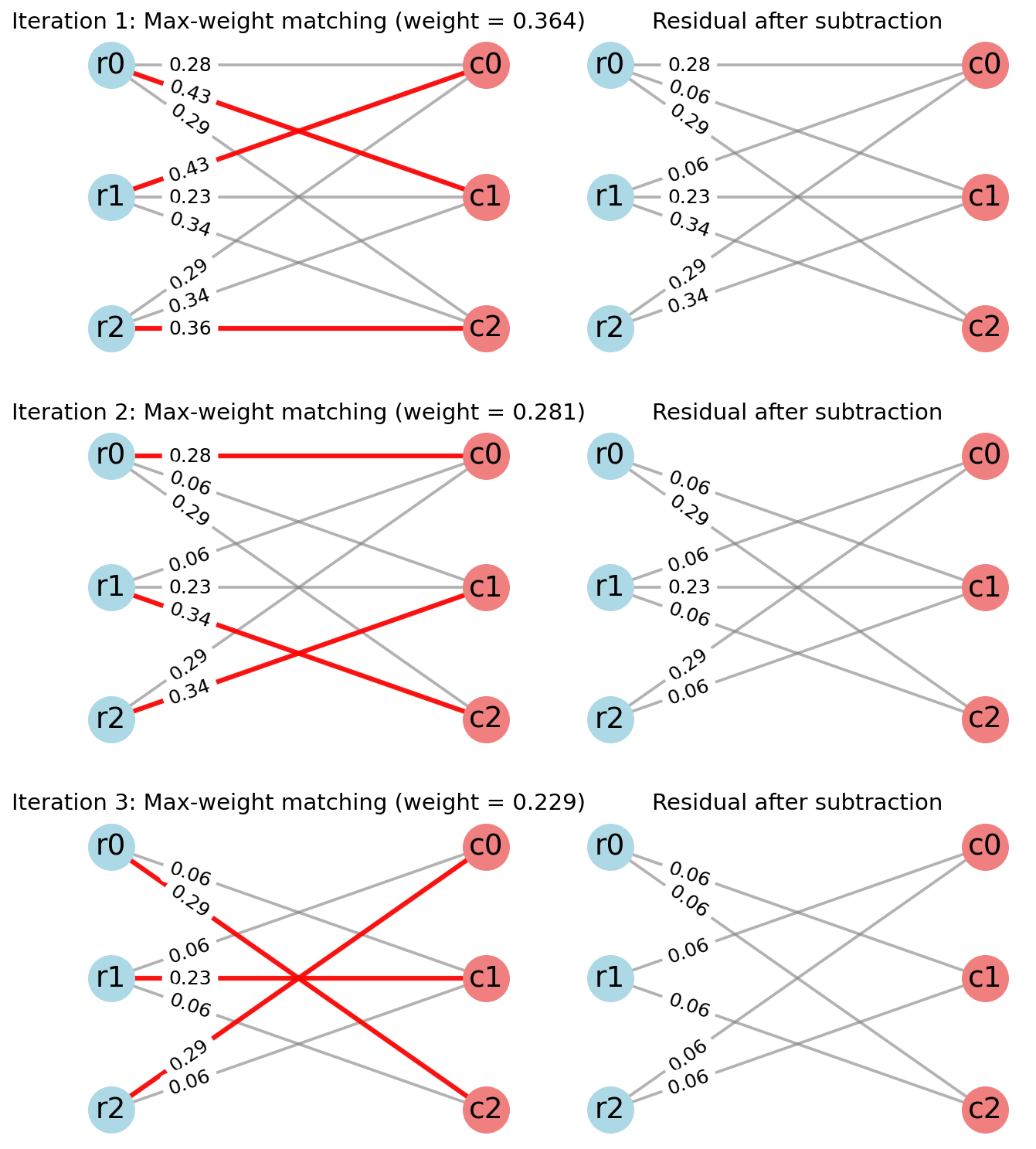}
\caption{Step‑by‑step execution of the largest‑weight BVN algorithm on a $3\times3$ doubly stochastic matrix. Left column: residual before subtraction with the chosen maximum‑weight perfect matching highlighted in red. Right column: residual after subtracting that matching.}
\label{fig:bvn-steps}
\end{figure*}

\subsection{Cut-off or threshold methods}
In Algorithm~\ref{alg:original}, before building the bipartite graph we can also remove all edges with $R_{ij} \le \theta$, where $\theta\ge0$ is a user-chosen threshold. This forces the algorithm to ignore very small entries and can drastically reduce $K$, at the cost of introducing an approximation error. The optimal $\theta$ (for a given error tolerance) can be found by binary search.

Or we can first compute the full original decomposition (Algorithm~\ref{alg:original} with $\varepsilon=0$). Sort the terms by weight and successively discard the smallest ones while the reconstruction error $\|S - \sum_{\text{kept}} w_k P_k\|_1$ stays below the given tolerance. Crucially, we \emph{do not} re-normalize the weights of the kept terms; this ensures the $\ell_1$ error is exactly equal to the sum of the discarded weights. This post-processing requires no extra matchings; its runtime is dominated by the initial $O(N^{4.5})$ decomposition.

\section{Theoretical Analysis}
\label{sec:complexity}

\subsection{Number of permutations}

Table~\ref{tab:termcounts} summarizes the term counts achieved by the
algorithms of Section~\ref{sec:bvn}.  Since the number of ancilla qubits
in an LCU block encoding is $\lceil\log_2 K\rceil$, the term count $K$ is
the decisive figure of merit for the quantum implementation.

\begin{table}[tb]
\centering
\caption{Worst‑case and typical number of permutations.
$N$ is the matrix size ($N=2^n$ for $n$ qubits).
The $O(N)$ scaling for the largest‑weight variant is conjectured
and supported by extensive numerical experiments.}
\label{tab:termcounts}
\resizebox{\columnwidth}{!}{%
\begin{tabular}{lcc}
\toprule
Method & $K$ (worst‑case) & $K$ (experiment) \\
\midrule
Original  & $N^2-2N+2$ \cite{brualdi1982notes} & $\Theta(N^2)$ \\
Cut‑off / Threshold              & $\le$ original            & $\le$ original   \\
Largest‑weight                    & $\mathcal{O}(N)^*$        & $\approx 2\,N$ \\
Bottleneck (theoretical)          & $N\ln(N/\varepsilon)$ (Thm.~\ref{thm:bottleneck-bvn}) & $N\ln(N/\varepsilon)$ \\
Pauli (generic) & $N^2$ & $N^2$ \\
\multicolumn{3}{l}{\footnotesize For symmetric matrices, the observed Pauli count is $N(N+1)/2$.}\\
\bottomrule
\end{tabular}
}
\end{table}

\textbf{Original BVN.}
The worst‑case number of permutations is at most $N^2-2N+2$, and this
bound is tight~\cite{brualdi1982notes,brualdi2006combinatorial}.
Finding the exact minimum is NP‑hard~\cite{dufosse2018further}, and
approximation schemes based on linear programming have been studied
\cite{kulkarni2017minimum,valls2021birkhoff}.
For a generic dense matrix the standard greedy algorithm produces
$\Theta(N^2)$ terms.

\textbf{Cut‑off and threshold variants.}
Both methods reduce the term count relative to the original decomposition
by discarding or ignoring small entries.  The cut‑off approach simply
prunes the smallest terms after a full decomposition, while the threshold
method avoids extracting very light matchings in the first place.
In our dense random tests with error tolerance $0.01$, the cut-off
method discarded approximately $20$--$25\%$ of the original terms for
$N\ge 16$; the threshold variant may produce larger reductions on
sparser matrices, but the worst-case bound remains $O(N^2)$.

\textbf{Largest‑weight BVN.}
For exact decomposition ($\varepsilon=0$) no worst‑case bound better
than the trivial $O(N^2)$ is known for the largest‑weight variant.
However, as we prove in the next subsection, a closely related
\emph{bottleneck} matching strategy already achieves
$O(N\log(N/\varepsilon))$ terms for an $\varepsilon$‑approximation.
In numerical experiments on dense doubly stochastic matrices
(Fig.~\ref{fig:experimental}) the largest‑weight variant
consistently reaches slightly higher than $K \approx 2\,N$, nearly independent of
$\varepsilon$ for moderate tolerances: 

The random doubly stochastic matrices were generated by drawing each entry uniformly from $[0.1,1]$, symmetrizing the matrix, and then applying Sinkhorn normalization until the row and column sums converged to $1$.  
For every dimension $N=2^n$, Fig.~\ref{fig:experimental} reports results from 50 independent random instances.  
Each marker is the mean over the 50 trials; error bars denote one standard deviation, and the individual trial points are shown as small semi-transparent markers with horizontal jitter solely to reduce overplotting.  

\begin{figure}[htbp]
\centering
 \includegraphics[width=1\linewidth]{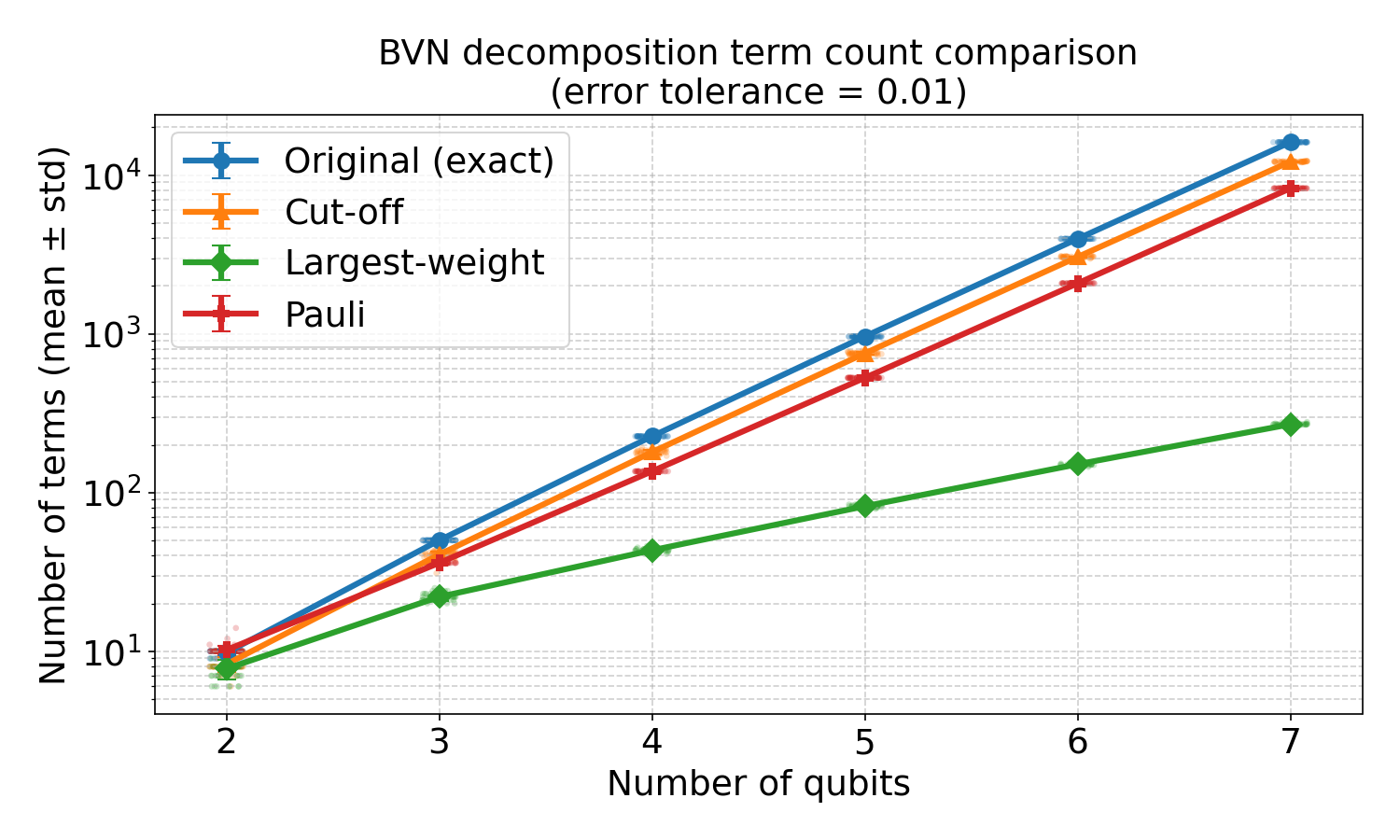}
\caption{{Empirical term counts for error tolerance $0.01$ and $N=2^n$ where $n$ represents number of qubits. Each marker is the mean over 50 random doubly stochastic matrices; error bars denote one standard deviation.  Small semi-transparent points show the individual 50 trials; the points are horizontally
jittered only to reduce over-plotting and do not affect the reported means or standard deviations (since standard deviations for Original and Pauli are almost 0 and for Largest-weight around 2-3, each trial value is concentrated around the mean and mostly not visible.).  The largest-weight variant scales as $O(N)$, while all other BVN methods and the Pauli decomposition
scale as $O(N^2)$.}}
\label{fig:experimental}
\end{figure}

As shown in the figure, the largest-weight term counts in our experiments remained around $\approx 2N$ for each trial, with standard deviations of only a few terms even as $N$ grows: The exact observed mean term counts of the largest-weight variant for $N=4,8,16,32,64,128$ are approximately $7.8$, $22.0$, $43.3$,
$81.9$, $150.7$, and $269$, respectively, with standard deviations $1.1$, $1.2$, $1.2$, $1.3$, $2.7$, and
$2.9$.  The resulting per-dimension scaling factors $K/N$ are approximately $1.95$, $2.75$, $2.71$, $2.56$, $2.35$, and $2.1$, giving an overall average scaling factor of approximately $2.4$.

We should also note that the term counts were essentially deterministic for the original BVN and for the Pauli decomposition with $N\ge 8$. In case of the Pauli decomposition, the observed term count was equal to $N(N+1)/2$which is different from the value given in Table~\ref{tab:termcounts}. This is expected because the random matrices are symmetric and  $N(N+1)/2$ is the number of independent symmetric real Pauli coefficients for a symmetric matrix.

The term count depends on the error tolerance: for precision $0.1$ the number of terms is around $10$, while for precision $0.0001$ it rises to approximately $150$; Fig.~\ref{fig:experimental}
corresponds to precision $0.01$. The dependence on precision is shown in Fig.~\ref{fig:precision},
where each point is also averaged over 50 trials.
\begin{figure}[t]
\centering
\begin{subfigure}[b]{0.48\textwidth}
     \includegraphics[width=\linewidth]{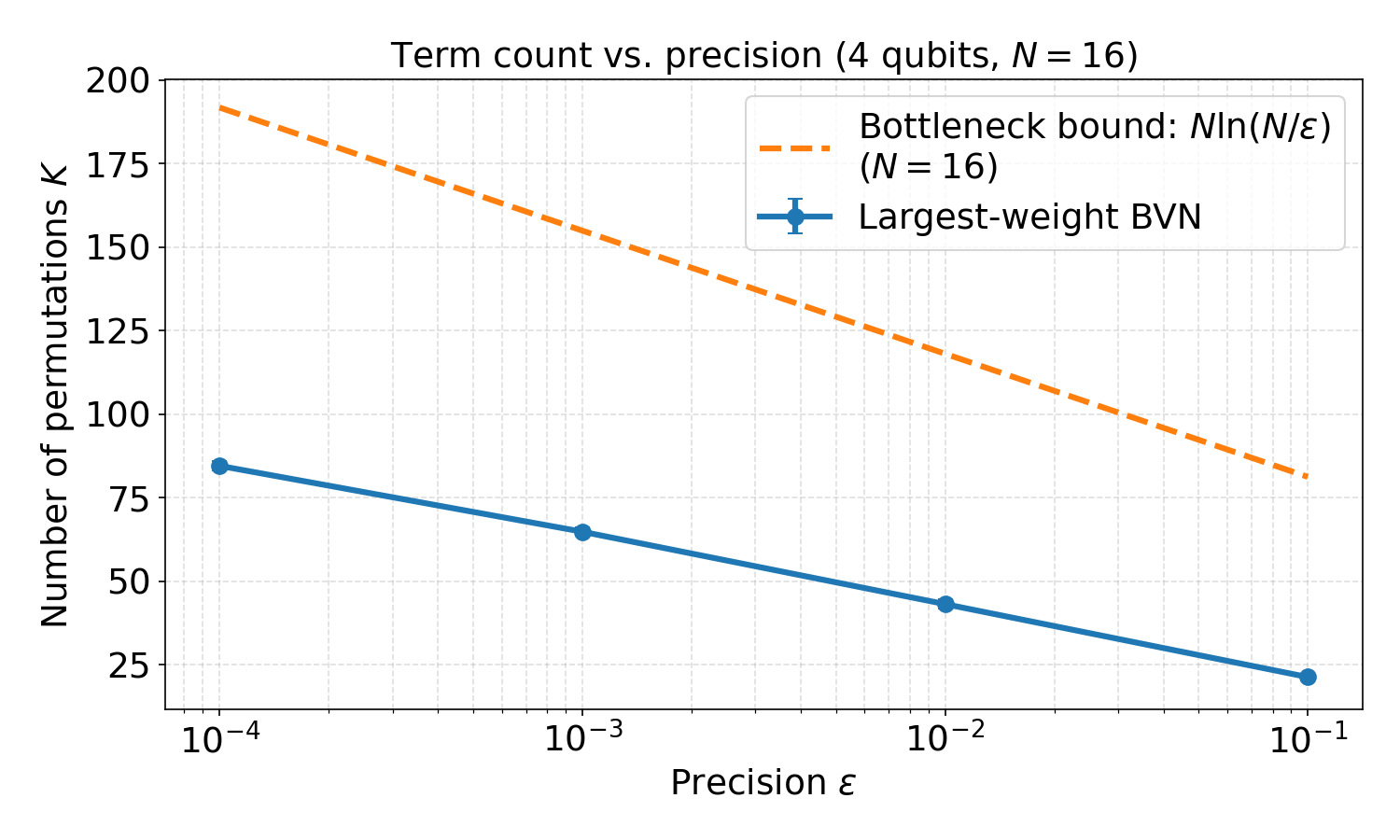}
\caption{Mean number of terms vs precision for 50 random $16\times16$ matrices.}
\end{subfigure}
\hfill 
\begin{subfigure}[b]{0.48\textwidth}
     \includegraphics[width=\linewidth]{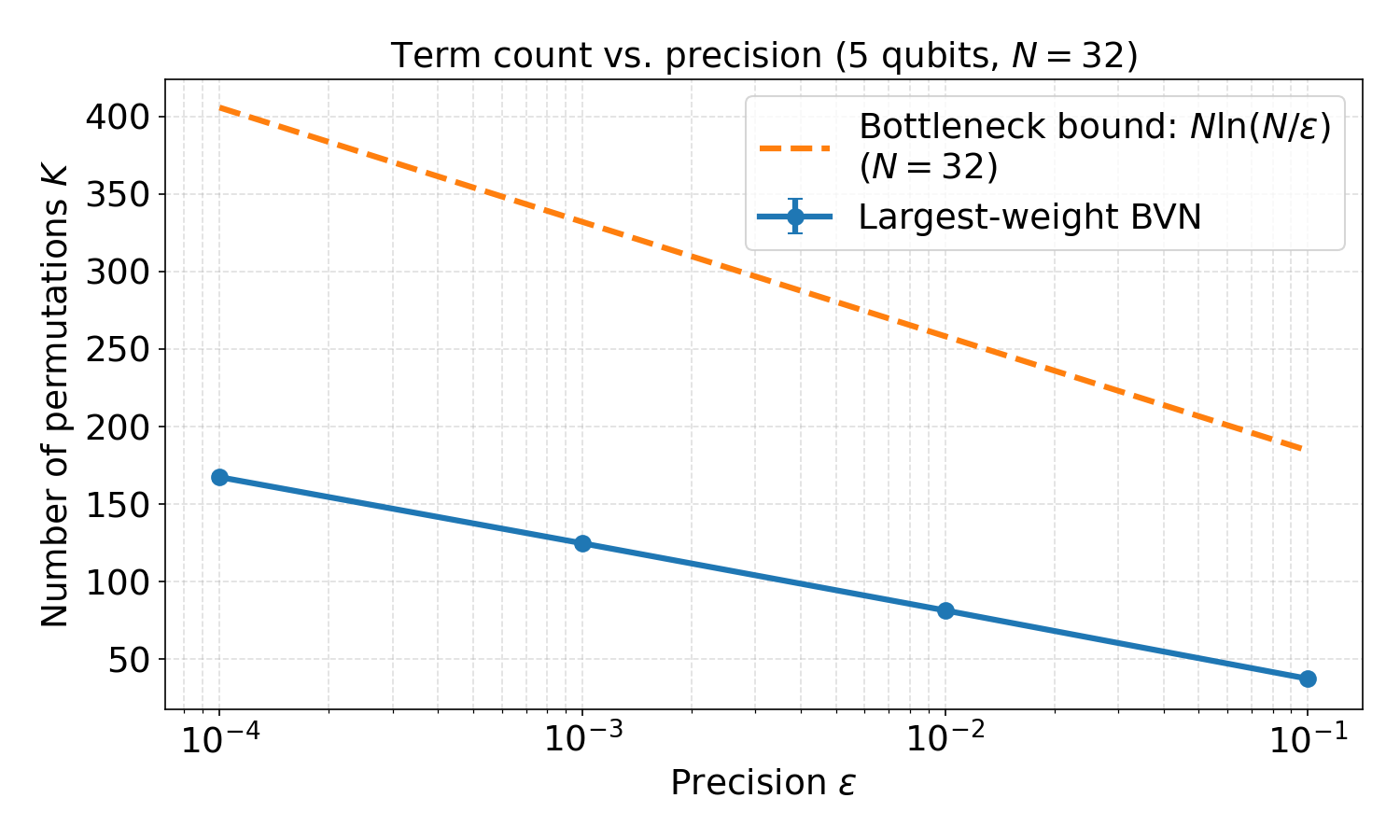}
\caption{Mean number of terms vs precision for 50 random $32\times32$ matrices.}
\end{subfigure}
\caption{Mean number of terms vs precision for 50 random matrices.}
\label{fig:precision}
\end{figure}
The small variance across trials at fixed $N$ is expected and is not an artifact of the random-matrix generator: for dense doubly stochastic matrices, the term count of the largest-weight greedy algorithm is a self-averaging quantity. After Sinkhorn normalization, different random instances share nearly identical residual mass profiles, so the maximum-weight matching extracts almost the same fraction of mass in every trial. As a result, the standard deviation of $K$ is typically only a few percent of the mean, and the 95\% bootstrap confidence intervals are narrower than the plotted marker size. This statistical concentration further supports the robustness of the reported $O(N)$ scaling.

In conclusion, this algorithm is a variant that realizes a full quadratic reduction in the term count relative to the original Birkhoff method, and we conjecture that an $O(N)$ bound
holds for almost all matrices in the relative interior of the Birkhoff polytope.

\textbf{Bottleneck variant guarantee.}
The bottleneck (maximin) perfect matching algorithm provides a rigorous upper bound of $O(N\log(N/\varepsilon))$ terms.  The following subsection presents the formal statement and proof.  While the bottleneck variant itself is not the most economical in practice, it demonstrates that matching‑based decompositions can, in principle, achieve near‑linear term counts---making the largest‑weight heuristic a natural and empirically superior choice.

\subsection{Theoretical guarantee for the bottleneck‑matching variant}
\label{sec:bottleneck-bound}
To gain insight into why the largest‑weight heuristic performs so well, we analyze a related algorithm that extracts a \emph{bottleneck} perfect matching at each step---i.e., a matching that maximizes the minimum edge weight.  
For doubly stochastic matrices the following classical result gives a lower bound on that minimum.

\begin{lemma}[Marcus and Ree~\cite{marcus1959diagonals}, see also \cite{KUSHWAHA2025223}]
\label{lem:heavy-permutation}
Let $R$ be a non‑negative $N\times N$ matrix whose row and column sums are all equal to $r>0$.  
Then there exists a permutation $\sigma$ such that $R_{i,\sigma(i)} \ge r/N$ for every $i=1,\dots,N$.  
Consequently, the bottleneck (maximin) perfect matching of the bipartite graph with edge weights $R_{ij}$ has value at least $r/N$.
\end{lemma}

\begin{proof}[Proof sketch]
The matrix $\frac{1}{r}R$ is doubly stochastic.  
The statement is a direct consequence of the classic theorem of Marcus and Ree~\cite{marcus1959diagonals}, which guarantees a permutation with all entries $\ge 1/N$ in any doubly stochastic matrix.
\end{proof}

Lemma~\ref{lem:heavy-permutation} immediately yields a complexity bound for the bottleneck variant of the Birkhoff--von~Neumann algorithm.

\begin{theorem}[Term count of bottleneck BVN]
\label{thm:bottleneck-bvn}
Let $S$ be an $N\times N$ doubly stochastic matrix and let $0<\varepsilon<1$.  The algorithm that repeatedly extracts a bottleneck perfect matching (maximin matching) and subtracts its minimum weight produces an $\varepsilon$‑approximation (in $\ell_1$ norm) after at most
\begin{equation}
K \;\le\; N\,\ln\frac{N}{\varepsilon}
\end{equation}
steps.
\end{theorem}
\begin{proof}
Let $r^{(t)}$ be the common row / column sum of the residual matrix after $t$ steps; initially $r^{(0)}\!=\!1$.  At step $t$ the bottleneck matching finds a perfect matching whose minimum weight is $w_t$.  By Lemma~\ref{lem:heavy-permutation} we have $w_t \ge r^{(t)}/N$.  The updated row sum therefore satisfies
\begin{equation}
r^{(t+1)} = r^{(t)} - w_t \;\le\; r^{(t)}\Bigl(1-\frac{1}{N}\Bigr).
\end{equation}
Iterating gives $r^{(K)} \le (1-1/N)^K$. Since $\ln(1-1/N) \le -1/N$, we need $K \ge N\,\ln(N/\varepsilon)$ to ensure $r^{(K)}\le\varepsilon/N$.  The $\ell_1$ norm of the residual equals $N\,r^{(K)}$, so the error is at most $N(\varepsilon/N) = \varepsilon$ after this many terms.
\end{proof}

Theorem~\ref{thm:bottleneck-bvn} shows that a bottleneck‑based decomposition already reduces the term count to $O(N\log(N/\varepsilon))$,in stark contrast to the $O(N^2)$ worst case of the original method.
In our experiments the \emph{largest‑weight} variant(Algorithm~\ref{alg:largest_weight}) consistently achieves even fewer terms---empirically $K \approx 2N$, nearly independent of $\varepsilon$ for moderate tolerances---because choosing a maximum‑weight matching removes a larger fraction of the residual mass at each step.  A tight theoretical analysis of the largest‑weight variant remains an open problem, but the practical gain is substantial (Fig.~\ref{fig:experimental}).

\subsection{Classical runtime of the decomposition algorithms}
\label{sec:classical-runtime}

The preprocessing time on a classical computer, while secondary for
quantum applications, determines the feasibility of the overall workflow.
All algorithms build and process a bipartite graph with $N$ vertices on
each side and up to $N^2$ edges.

\begin{itemize}
\item \textbf{Original (any perfect matching).}  A single perfect matching
  via Hopcroft–Karp costs $O(N^{2.5})$.  With $O(N^2)$ iterations the
  total time is $O(N^{4.5})$.
\item \textbf{Cut‑off and threshold.}  The same matching routine is used;
  the cut‑off approach incurs an additional $O(N^3)$ pruning step.
  Worst‑case time remains $O(N^{4.5})$.
\item \textbf{Largest‑weight and bottleneck.}  Both require a\emph{maximum‑weight perfect matching} per iteration.  In our Python implementation we employ NetworkX’s \texttt{max\_weight\_matching}, which is based on Edmonds’s blossom algorithm~\cite{edmonds1965paths} and runs in $O(N^3)$ per call. 
With $K = O(N)$ iterations this gives $O(N^4)$ total time.Faster maximum‑weight matching algorithms exist: 
the Gabow–Tarjan algorithm~\cite{gabow1974implementation} also achieves $O(N^3)$, and Duan and Pettie~\cite{duan2014linear} give a linear‑time$(1-\delta)$‑approximation for any constant $\delta>0$.  
The latter would reduce the per‑iteration cost to $O(N^2)$ and the total to$O(N^3)$, with negligible impact on the final permutation count.
\end{itemize}

The constant factors hidden in the asymptotic notation are small enough that matrices with $N \le 10^4$ can be decomposed on a standard laptop in seconds to minutes.  For larger sizes, the approximate linear‑time algorithms provide a practical path to scaling.

\section{Discussion}
\label{sec:discussion}

\subsection{Quantum resource implications of fewer terms}
The primary contribution of this work is a demonstration that the largest‑weight greedy 
Birkhoff--von~Neumann decomposition can reduce the number of permutation terms from $\Theta(N^2)$ (for standard Birkhoff or Pauli decompositions) to $\approx 2N$ for dense doubly stochastic matrices of size $N = 2^n$.  This has several direct consequences for the block‑encoding of such matrices:

\begin{itemize}
\item \textbf{Ancilla qubits.} The ancilla register requires $\lceil\log_2 K\rceil$ qubits, so the reduction from $O(N^2)$ to $O(N)$ terms decreases the ancilla width from $\sim 2\log_2 N$ to $\sim \log_2 N$ --- a quadratic improvement in the number of ancillas. For near‑term devices with limited qubit counts, this can be the difference between feasibility and infeasibility. 

\item \textbf{Success probability.}  Because the decomposition is a convex combination of permutation matrices, the LCU normalization constant is exactly $\alpha = 1$.  The success probability of the ancilla post‑selection is therefore $\|S\ket{\psi}\|^2$, with no hidden scaling penalty.  Crucially, the uniform superposition$\ket{+}^{\otimes n}$ is a known eigenvector of $S$ with eigenvalue $1$,yielding \emph{deterministic} success ($p_{\text{succ}} = 1$) when the block encoding is applied to this input.  For other inputs the success probability remains $O(1)$ as long as the state overlaps significantly with the principal subspace.  This stands in sharp contrast to a Pauli decomposition, where $\alpha$ is often $\Theta(\sqrt{N})$ or larger and its exact value is unknown without explicit calculation, forcing the use of amplitude amplification. 
{In fixed‑Hadamard LCU architectures \cite{daskin2026exploiting},
where---for uniformly weighted LCU coefficients---the success
probability scales as $1/K$, our reduction in $K$ further amplifies
this advantage.  In contrast, our convex decomposition has
$\alpha=1$ exactly, so the block encoding does not rely on this
uniform-coefficient assumption and can achieve deterministic success
on the uniform eigenstate.}

\item \textbf{Coefficient preparation and SELECT complexity.}  The
classical preprocessing of the $K$ coefficients and the quantum SELECT
oracle (multiplexed controlled‑permutation) both scale linearly with
$K$.  Hence shrinking $K$ by a factor $\sim N$ directly reduces the
circuit depth and the number of controlled gates, lowering the overall
error susceptibility on noisy hardware.
\end{itemize}

\subsection{Trade‑off between term count and CNOT gate count}

It is important to note that a smaller $K$ does \emph{not} automatically guarantee a lower CNOT count.  While a Pauli term is implemented by a single multi‑controlled Pauli gate (or a few single‑qubit rotations), 
{a single $n$‑qubit permutation may require up to $\Theta(n\,2^n)$ elementary
gates in the worst case~\cite{shende2003synthesis,saeedi2013synthesis}.}  For instance,
implementing a general $n$‑qubit permutation using only Toffoli and CNOT
gates can cost $\mathcal{O}(n\,2^n)$ gates~\cite{shende2003synthesis,saeedi2013synthesis}.
Thus, even with a linear number of terms, the overall CNOT count might
still scale as $\Theta(N^2 \log N)$ in the worst case, comparable to the
Pauli approach.

However, several factors mitigate this concern in practice:
(i) Many permutations are relatively
simple to implement on quantum computers: e.g., they correspond to permutations of computational‑basis
states that can be implemented by a small circuit;
(ii) the bound of $O(n\,2^n)$ is worst case \cite{saeedi2013synthesis} and rarely reached for structured permutations;
(iii) the dominant cost in LCU is often the SELECT operation rather than the controlled unitaries themselves, and a smaller $K$ directly reduces the depth of the SELECT oracle.
Ultimately, the choice between permutation‑based and Pauli‑based block‑encoding depends on the target hardware and the specific matrix; our largest‑weight BVN provides a new option with a qualitatively different resource profile.

\textbf{Quantitative crossover estimate.}
To make the trade-off explicit, let $N=2^n$ and model the number of
terms as $K_{\rm Pauli}=N^2$ and $K_{\rm perm}\approx 2.4N$.
A single controlled Pauli term in an LCU SELECT circuit can be
implemented with $O(n)$ CNOTs; we take the order-of-magnitude
expression $a\,n$ with $a=2$.  A general $n$-qubit permutation
requires at most $O(nN)$ CNOTs~\cite{shende2003synthesis}; we model
this as $b\,nN$ with $b=3$.  For structured permutations arising from
local permutations, bit-reversals, or short cycles, a controlled
permutation can often be implemented with $O(n^2)$ CNOTs, which we
model as $c\,n^2$ with $c=1$.

Table~\ref{tab:cnot-crossover} shows the resulting order-of-magnitude
total SELECT CNOT counts for the bit-permutation component.  Under worst-case permutation synthesis, the
total CNOT scaling remains $O(N^2 n)$, and the term-count reduction is
offset by the per-permutation synthesis cost; in this regime the
primary advantage is ancilla width, not total CNOT count.  For
structured permutations, however, the total CNOT scales only as
$O(N n^2)$, yielding a large improvement already at moderate $N$.
The table is intended as an order-of-magnitude comparison, not an
exact gate count.

Crucially, for general complex matrices, each phased permutation matrix $P_k$ contains entry-dependent phases $e^{i\theta_{ij}}$. Implementing a controlled phased permutation requires not only the bit-permutation circuit but also a phase oracle that applies these $N$ phases. For arbitrary phases, this requires loading $N$ angles per term, typically costing $O(N)$ gates using QROM or amplitude loading. Consequently, for general complex matrices, the per-term cost increases by $O(N)$, leading to a total SELECT cost of $O(K N) \approx O(N^2)$. In this regime, the CNOT count advantage over Pauli decompositions may diminish, but the reduction in ancilla width from $2\log_2 N$ to $\log_2 N$ remains a significant hardware advantage. For matrices where the phases are structured or sparse, the phase synthesis cost can be substantially lower.

\begin{table}[hptb]
\centering
\footnotesize
\caption{Illustrative SELECT CNOT estimates for LCU block encodings of \emph{standard} (non-phased) doubly stochastic matrices.
$K$ is the number of terms; ``CNOT/term'' is an order-of-magnitude
per-term synthesis cost for the bit-permutation component; total CNOT is $K\times$ CNOT/term.
Constants: $a=2$, $b=3$, $c=1$. For $N=16,32,64,128$, $K$ is the measured mean largest-weight term count from 50 random symmetric doubly stochastic matrices; for $N=1024$, $K$ is extrapolated as $2.4N$. \emph{Note:} For general complex matrices, an additional $O(N)$ phase synthesis cost per term applies, which may offset the CNOT advantage but preserves the ancilla width reduction.}
\label{tab:cnot-crossover}
\resizebox{\columnwidth}{!}{%
\begin{tabular}{lllll}
\toprule
$N$ & Method & $K$ & CNOT/term & Total CNOT \\
\midrule
16  & Pauli                & 136    & $a n=8$          & $\approx 1.1\times10^3$ \\
16  & Perm. worst-case     & 43     & $b n N=192$      & $\approx 8.3\times10^3$ \\
16  & Perm. structured     & 43     & $c n^2=16$       & $\approx 6.9\times10^2$ \\
\midrule
32  & Pauli                & 528    & $a n=10$         & $\approx 5.3\times10^3$ \\
32  & Perm. worst-case     & 82     & $b n N=480$      & $\approx 3.9\times10^4$ \\
32  & Perm. structured     & 82     & $c n^2=25$       & $\approx 2.1\times10^3$ \\
\midrule
64  & Pauli                & 2080   & $a n=12$         & $\approx 2.5\times10^4$ \\
64  & Perm. worst-case     & 151    & $b n N=1152$     & $\approx 1.7\times10^5$ \\
64  & Perm. structured     & 151    & $c n^2=36$       & $\approx 5.4\times10^3$ \\
\midrule
128 & Pauli                & 8256   & $a n=14$         & $\approx 1.2\times10^5$ \\
128 & Perm. worst-case     & 269    & $b n N=2688$     & $\approx 7.2\times10^5$ \\
128 & Perm. structured     & 269    & $c n^2=49$       & $\approx 1.3\times10^4$ \\
\midrule
1024 & Pauli               & $524{,}800$ & $a n=20$    & $\approx 1.0\times10^7$ \\
1024 & Perm. worst-case    & 2458   & $b n N=30720$    & $\approx 7.6\times10^7$ \\
1024 & Perm. structured    & 2458   & $c n^2=100$      & $\approx 2.5\times10^5$ \\
\bottomrule
\end{tabular}
}
\end{table}

\subsection{Future directions} 
We believe several extensions of this work are worth pursuing: while our phased Birkhoff--von Neumann framework accommodates general complex entries by embedding phases directly into the permutation matrices, future work could investigate whether a joint optimization of phases and permutation weights could yield even sparser representations. Additionally, developing a specialized bottleneck matching heuristic that directly handles complex phases during the classical bipartite matching stage could provide tighter theoretical bounds on the term count for specific operator classes, such as non-Hermitian Hamiltonians or quantum channels.

The block‑encoding circuit ultimately depends on the quality of the
controlled‑permutation implementation.  Adapting the permutation
synthesis to the native gate set of a specific quantum computer (e.g.,
incorporating nearest‑neighbor constraints) may significantly reduce the CNOT overhead.

Finally, as stated before, the method is generic: {That means any complex square matrix whose absolute profile can be Sinkhorn‑scaled benefits directly from our generalized phased permutation‑based block‑encoding framework.} This opens up applications in non‑Hermitian simulation,
quantum reinforcement learning, and wherever a dense linear operator
needs to be embedded into a unitary.

\section{Conclusion}

We have shown that the largest‑weight greedy Birkhoff–von Neumann
decomposition, a simple heuristic from classical scheduling, is remarkably
effective for building block encodings of doubly stochastic matrices.
On dense instances it produces only $\approx 2N$ permutation terms,
compared to $\Theta(N^2)$ for standard Birkhoff or Pauli decompositions.
This quadratic reduction in the number of terms translates directly into
fewer ancilla qubits, shallower SELECT circuits, and, for some LCU
architectures, higher success probability.  While the classical
pre‑processing cost is $O(N^4)$, it remains manageable for
problem sizes where quantum advantage is expected, and approximate
linear‑time matching algorithms can lower this cost further.

The permutation‑based LCU approach thus opens a practical path toward
quantum algorithms for optimal transport, Sinkhorn‑scaled linear operators,
and other dense‑matrix applications.  We anticipate that the combination of {phased Sinkhorn scaling} and largest‑weight Birkhoff decomposition  will become
a standard tool in the quantum linear algebra toolbox.

\section*{Acknowledgments}
The author acknowledges the use of AI assistant tools during the preparation of this manuscript. 

\section*{Funding}
This project is not funded by any funding agency.

\section*{Data Availability}
The simulation code can be downloaded from GitHub at \url{https://github.com/adaskin/lcu-birkhoff} where we also saved the Jupyter-notebooks that generated figures.
\bibliographystyle{ieeetr}
\bibliography{main}

@article{berry_simulating_2015,
	title = {Simulating Hamiltonian Dynamics with a Truncated Taylor Series},
	volume = {114},
	rights = {http://link.aps.org/licenses/aps-default-license},
	issn = {0031-9007, 1079-7114},
	url = {https://link.aps.org/doi/10.1103/PhysRevLett.114.090502},
	doi = {10.1103/PhysRevLett.114.090502},
	pages = {090502},
	number = {9},
	journal = {Physical Review Letters},
	shortjournal = {Phys. Rev. Lett.},
	author = {Berry, Dominic W. and Childs, Andrew M. and Cleve, Richard and Kothari, Robin and Somma, Rolando D.},
	date = {2015-03-03}, 
    year="2015",
	langid = {english}
}

@article{sinkhorn1967concerning,
  title={Concerning nonnegative matrices and doubly stochastic matrices},
  author={Sinkhorn, Richard and Knopp, Paul},
  journal={Pacific Journal of Mathematics},
  volume={21},
  number={2},
  pages={343--348},
  year={1967},
  publisher={Mathematical Sciences Publishers}
}

@article{berry2007efficient,
  title={Efficient quantum algorithms for simulating sparse Hamiltonians},
  author={Berry, Dominic W and Ahokas, Graeme and Cleve, Richard and Sanders, Barry C},
  journal={Communications in Mathematical Physics},
  volume={270},
  number={2},
  pages={359--371},
  year={2007},
  publisher={Springer}
}

@article{pereira2014theory,
  title={The theory and applications of complex matrix scalings},
  author={Pereira, Rajesh and Boneng, Joanna},
  journal={Special Matrices},
  volume={2},
  number={1},
  pages={68--77},
  year={2014}
}

@article{hutchinson2023scaling,
  title={SCALING POSITIVE DEFINITE MATRICES TO ACHIEVE PRESCRIBED EIGENPAIRS},
  author={Hutchinson, George},
  journal={OPERATORS AND MATRICES},
  volume={17},
  number={4},
  pages={967--994},
  year={2023},
  publisher={ELEMENT R AUSTRIJE 11, 10000 ZAGREB, CROATIA}
}

@article{idel2016review,
  title={A review of matrix scaling and Sinkhorn's normal form for matrices and positive maps},
  author={Idel, Martin},
  journal={arXiv preprint arXiv:1609.06349},
  year={2016}
}

@article{edmonds1965paths,
  title={Paths, trees, and flowers},
  author={Edmonds, Jack},
  journal={Canadian Journal of mathematics},
  volume={17},
  pages={449--467},
  year={1965},
  publisher={Cambridge University Press}
}

@article{de2019birkhoff,
  title={The Birkhoff theorem for unitary matrices of prime-power dimension},
  author={De Vos, Alexis and De Baerdemacker, Stijn},
  journal={Linear Algebra and its Applications},
  volume={578},
  pages={27--52},
  year={2019},
  publisher={Elsevier}
}

@article{welch2014efficient,
  title={Efficient quantum circuits for diagonal unitaries without ancillas},
  author={Welch, Jonathan and Greenbaum, Daniel and Mostame, Sarah and Aspuru-Guzik, Alan},
  journal={New Journal of Physics},
  volume={16},
  number={3},
  pages={033040},
  year={2014},
  publisher={IOP Publishing}
}

@article{shende2003synthesis,
  title={Synthesis of reversible logic circuits},
  author={Shende, Vivek V and Prasad, Aditya K and Markov, Igor L and Hayes, John P},
  journal={IEEE Transactions on Computer-Aided Design of Integrated Circuits and Systems},
  volume={22},
  number={6},
  pages={710--722},
  year={2003},
  publisher={IEEE}
}

@article{saeedi2013synthesis,
  title={Synthesis and optimization of reversible circuits—a survey},
  author={Saeedi, Mehdi and Markov, Igor L},
  journal={ACM Computing Surveys (CSUR)},
  volume={45},
  number={2},
  pages={1--34},
  year={2013},
  publisher={ACM New York, NY, USA}
}

@article{daskin2026exploiting,
  title={Exploiting all ancilla outcomes in linear combinations of unitaries: low-rank recovery and quantum trapdoor functions},
  author={Daskin, Ammar},
  journal={arXiv preprint arXiv:2605.02986},
  year={2026}
}

@inproceedings{gilyen_quantum_2019,
	location = {Phoenix {AZ} {USA}},
	title = {Quantum singular value transformation and beyond: exponential improvements for quantum matrix arithmetics},
	isbn = {978-1-4503-6705-9},
	url = {https://dl.acm.org/doi/10.1145/3313276.3316366},
	doi = {10.1145/3313276.3316366},
	shorttitle = {Quantum singular value transformation and beyond},
	eventtitle = {{STOC} '19: 51st Annual {ACM} {SIGACT} Symposium on the Theory of Computing},
	pages = {193--204},
	booktitle = {Proceedings of the 51st Annual {ACM} {SIGACT} Symposium on Theory of Computing},
	publisher = {{ACM}},
	author = {Gilyén, András and Su, Yuan and Low, Guang Hao and Wiebe, Nathan},
	date = {2019-06-23},
    year="2019",
	langid = {english}
}

@article{martyn_grand_2021,
	title = {Grand Unification of Quantum Algorithms},
	volume = {2},
	issn = {2691-3399},
	url = {https://link.aps.org/doi/10.1103/PRXQuantum.2.040203},
	doi = {10.1103/PRXQuantum.2.040203},
	pages = {040203},
	number = {4},
	journal= {{PRX} Quantum},
	shortjournal = {{PRX} Quantum},
	author = {Martyn, John M. and Rossi, Zane M. and Tan, Andrew K. and Chuang, Isaac L.},
	date = {2021-12-03},
    year="2021",
	langid = {english}
}

@article{bharadwaj2025compact,
  title={Compact quantum algorithms for time-dependent differential equations},
  author={Bharadwaj, Sachin S and Sreenivasan, Katepalli R},
  journal={Physical Review Research},
  volume={7},
  number={2},
  pages={023262},
  year={2025},
  publisher={APS}
}

@article{gharat2026quantum,
  title={Quantum algorithm for solving differential equations using SLAC derivatives},
  author={Gharat, Rakshit M and Muraleedharan, Gopikrishnan and Berry, Dominic W and Brennen, Gavin K},
  journal={arXiv preprint arXiv:2605.04861},
  year={2026}
}

@article{xin2020quantum,
  title={Quantum algorithm for solving linear differential equations: Theory and experiment},
  author={Xin, Tao and Wei, Shijie and Cui, Jianlian and Xiao, Junxiang and Arrazola, I{\~n}igo and Lamata, Lucas and Kong, Xiangyu and Lu, Dawei and Solano, Enrique and Long, Guilu},
  journal={Physical Review A},
  volume={101},
  number={3},
  pages={032307},
  year={2020},
  publisher={APS}
}

@article{daskin2024quantum,
  title={A quantum compiler design method by using linear combinations of permutations},
  author={Daskin, Ammar},
  journal={arXiv preprint arXiv:2404.18226},
  year={2024}
}

@article{heredge2025nonunitary,
  title={Nonunitary quantum machine learning},
  author={Heredge, Jamie and West, Maxwell and Hollenberg, Lloyd and Sevior, Martin},
  journal={Physical Review Applied},
  volume={23},
  number={4},
  pages={044046},
  year={2025},
  publisher={APS}
}

@article{garcia2025nuclear,
  title={Nuclear physics in the era of quantum computing and quantum machine learning},
  author={Garc{\'\i}a-Ramos, Jos{\'e}-Enrique and S{\'a}iz, {\'A}lvaro and Arias, Jos{\'e} M and Lamata, Lucas and P{\'e}rez-Fern{\'a}ndez, Pedro},
  journal={Advanced Quantum Technologies},
  volume={8},
  number={12},
  pages={2300219},
  year={2025},
  publisher={Wiley Online Library}
}

@article{zhai2026quantum,
  title={Quantum Neural Physics: Solving Partial Differential Equations on Quantum Simulators using Quantum Convolutional Neural Networks},
  author={Zhai, Jucai and Abdullah, Muhammad and Chen, Boyang and Chaudry, Fazal and Smith, Paul N and Heaney, Claire E and Wang, Yanghua and Xiang, Jiansheng and Pain, Christopher C},
  journal={arXiv preprint arXiv:2603.24196},
  year={2026}
}

@article{gnanasekaran_efficient_2026,
	title = {Efficient quantum access model for sparse structured matrices using linear combination of “things”},
	volume = {113},
	issn = {2469-9926, 2469-9934},
	url = {https://link.aps.org/doi/10.1103/x4b2-9d94},
	doi = {10.1103/x4b2-9d94},
	pages = {022437},
	number = {2},
	journal = {Physical Review A},
	shortjournal = {Phys. Rev. A},
	author = {Gnanasekaran, Abeynaya and Surana, Amit},
	date = {2026-02-23},
    year="2026",
	langid = {english}
}

@article{camps_explicit_2024,
	title = {Explicit Quantum Circuits for Block Encodings of Certain Sparse Matrices},
	volume = {45},
	issn = {0895-4798, 1095-7162},
	url = {https://epubs.siam.org/doi/10.1137/22M1484298},
	doi = {10.1137/22M1484298},
	pages = {801--827},
	number = {1},
	journal = {{SIAM} Journal on Matrix Analysis and Applications},
	shortjournal = {{SIAM} J. Matrix Anal. Appl.},
	author = {Camps, Daan and Lin, Lin and Van Beeumen, Roel and Yang, Chao},
	date = {2024-03-31},
    year="2024",
	langid = {english}
}

@inproceedings{szegedy_quantum_2004,
	title = {Quantum speed-up of Markov chain based algorithms},
	url = {https://ieeexplore.ieee.org/abstract/document/1366222/},
	pages = {32--41},
	booktitle = {45th Annual {IEEE} symposium on foundations of computer science},
	publisher = {{IEEE}},
	author = {Szegedy, Mario},
	year = {2004},
}

@article{della2025efficient,
  title={Efficient LCU block encodings through Dicke states preparation},
  author={Della Chiara, Filippo and Nibbi, Martina and Shen, Yizhi and Van Beeumen, Roel},
  journal={arXiv preprint arXiv:2507.20887},
  year={2025}
}

@article{loaiza2025majorana,
  title={Majorana Tensor Decomposition: A unifying framework for decompositions of fermionic Hamiltonians to Linear Combination of Unitaries},
  author={Loaiza, Ignacio and Sankar Brahmachari, Aritra and Izmaylov, Artur F},
  journal={Quantum Science and Technology},
  volume={10},
  number={3},
  pages={035035},
  year={2025},
  publisher={IOP Publishing}
}

@article{orts_entropic_2026,
	title = {Entropic optimal transport with quantum amplitude estimation},
	url = {https://www.sciencedirect.com/science/article/pii/S0167739X26002049},
	pages = {108570},
	journal = {Future Generation Computer Systems},
	publisher = {Elsevier},
	author = {Orts, Francisco},
	date = {2026},
    year="2026",
}

@article{daskin_error_2025,
	title = {Error analysis of quantum operators written as a linear combination of permutations},
	volume = {24},
	issn = {1573-1332},
	url = {https://link.springer.com/10.1007/s11128-025-04771-0},
	doi = {10.1007/s11128-025-04771-0},
	pages = {149},
	number = {5},
	journal = {Quantum Information Processing},
	shortjournal = {Quantum Inf Process},
	author = {Daskin, Ammar},
	date = {2025-05-22},
    year="2025",
	langid = {english}
}

@article{lapworth_preconditioned_2025,
	title = {Preconditioned block encodings for quantum linear systems},
	volume = {10},
	url = {https://iopscience.iop.org/article/10.1088/2058-9565/ae0f4b/meta},
	pages = {045064},
	number = {4},
	journal = {Quantum Science and Technology},
	publisher = {{IOP} Publishing},
	author = {Lapworth, Leigh and Sünderhauf, Christoph},
	date = {2025},
    year="2025",
}

@article{childs2012hamiltonian,
  title={Hamiltonian simulation using linear combinations of unitary operations},
  author={Childs, Andrew M and Wiebe, Nathan},
  journal={Quantum Information \& Computation},
  volume={12},
  number={11-12},
  pages={901--924},
  year={2012},
  publisher={Rinton Press, Incorporated Paramus, NJ}
}

@article{duan2014linear,
  title={Linear-time approximation for maximum weight matching},
  author={Duan, Ran and Pettie, Seth},
  journal={Journal of the ACM (JACM)},
  volume={61},
  number={1},
  pages={1--23},
  year={2014},
  publisher={ACM New York, NY, USA}
}

@book{gabow1974implementation,
  title={Implementation of algorithms for maximum matching on nonbipartite graphs.},
  author={Gabow, Harold Neil},
  year={1974},
  publisher={Stanford University}
}

@article{hopcroft1973n,
  title={An n\^{}5/2 algorithm for maximum matchings in bipartite graphs},
  author={Hopcroft, John E and Karp, Richard M},
  journal={SIAM Journal on computing},
  volume={2},
  number={4},
  pages={225--231},
  year={1973},
  publisher={SIAM}
}

@article{marcus1959diagonals,
  title={Diagonals of doubly stochastic matrices},
  author={Marcus, Marvin and Ree, Rimhak},
  journal={The Quarterly Journal of Mathematics},
  volume={10},
  number={1},
  pages={296--302},
  year={1959},
  publisher={Oxford University Press}
}

@article{KUSHWAHA2025223,
title = {A note on Erdős matrices and Marcus–Ree inequality},
journal = {Linear Algebra and its Applications},
volume = {725},
pages = {223-247},
year = {2025},
issn = {0024-3795},
doi = {https://doi.org/10.1016/j.laa.2025.07.012},
url = {https://www.sciencedirect.com/science/article/pii/S0024379525002964},
author = {Aman Kushwaha and Raghavendra Tripathi}

}

@article{dufosse2018further,
  title={Further notes on Birkhoff--von Neumann decomposition of doubly stochastic matrices},
  author={Dufoss{\'e}, Fanny and Kaya, Kamer and Panagiotas, Ioannis and U{\c{c}}ar, Bora},
  journal={Linear Algebra and its Applications},
  volume={554},
  pages={68--78},
  year={2018},
  publisher={Elsevier}
}

@inproceedings{kulkarni2017minimum,
  title={Minimum birkhoff-von neumann decomposition},
  author={Kulkarni, Janardhan and Lee, Euiwoong and Singh, Mohit},
  booktitle={International Conference on Integer Programming and Combinatorial Optimization},
  pages={343--354},
  year={2017},
  organization={Springer}
}

@book{brualdi2006combinatorial,
  title={Combinatorial matrix classes},
  author={Brualdi, Richard A},
  volume={13},
  year={2006},
  publisher={Cambridge University Press}
}

@article{valls2021birkhoff,
  title={Birkhoff’s decomposition revisited: Sparse scheduling for high-speed circuit switches},
  author={Valls, V{\'\i}ctor and Iosifidis, George and Tassiulas, Leandros},
  journal={IEEE/ACM Transactions on Networking},
  volume={29},
  number={6},
  pages={2399--2412},
  year={2021},
  publisher={IEEE}
}

@article{brualdi1982notes,
  title={Notes on the Birkhoff algorithm for doubly stochastic matrices},
  author={Brualdi, Richard A},
  journal={Canadian Mathematical Bulletin},
  volume={25},
  number={2},
  pages={191--199},
  year={1982},
  publisher={Cambridge University Press}
}

@inproceedings{liu2015scheduling,
  title={Scheduling techniques for hybrid circuit/packet networks},
  author={Liu, He and Mukerjee, Matthew K and Li, Conglong and Feltman, Nicolas and Papen, George and Savage, Stefan and Seshan, Srinivasan and Voelker, Geoffrey M and Andersen, David G and Kaminsky, Michael and others},
  booktitle={Proceedings of the 11th ACM Conference on Emerging Networking Experiments and Technologies},
  pages={1--13},
  year={2015}
}

@article{bojja2018costly,
  title={Costly circuits, submodular schedules and approximate carath{\'e}odory theorems},
  author={Bojja Venkatakrishnan, Shaileshh and Alizadeh, Mohammad and Viswanath, Pramod},
  journal={Queueing Systems},
  volume={88},
  number={3},
  pages={311--347},
  year={2018},
  publisher={Springer}
}

@article{babbush2018encoding,
  title={Encoding electronic spectra in quantum circuits with linear T complexity},
  author={Babbush, Ryan and Gidney, Craig and Berry, Dominic W and Wiebe, Nathan and McClean, Jarrod and Paler, Alexandru and Fowler, Austin and Neven, Hartmut},
  journal={Physical Review X},
  volume={8},
  number={4},
  pages={041015},
  year={2018},
  publisher={APS}
}

@article{low2019hamiltonian,
  title={Hamiltonian simulation by qubitization},
  author={Low, Guang Hao and Chuang, Isaac L},
  journal={Quantum},
  volume={3},
  pages={163},
  year={2019},
  publisher={Verein zur F{\"o}rderung des Open Access Publizierens in den Quantenwissenschaften}
}

@article{berry2019qubitization,
  title={Qubitization of arbitrary basis quantum chemistry leveraging sparsity and low rank factorization},
  author={Berry, Dominic W and Gidney, Craig and Motta, Mario and McClean, Jarrod R and Babbush, Ryan},
  journal={Quantum},
  volume={3},
  pages={208},
  year={2019},
  publisher={Verein zur F{\"o}rderung des Open Access Publizierens in den Quantenwissenschaften}
}

@article{low2024trading,
  title={Trading T gates for dirty qubits in state preparation and unitary synthesis},
  author={Low, Guang Hao and Kliuchnikov, Vadym and Schaeffer, Luke},
  journal={Quantum},
  volume={8},
  pages={1375},
  year={2024},
  publisher={Verein zur F{\"o}rderung des Open Access Publizierens in den Quantenwissenschaften}
}

@article{grover2002creating,
  title={Creating superpositions that correspond to efficiently integrable probability distributions},
  author={Grover, Lov and Rudolph, Terry},
  journal={arXiv preprint quant-ph/0208112},
  year={2002}
}

@misc{GuillermoAlonso2024,
    title = "Intro to quantum read-only memory (QROM)",
    author = "Guillermo Alonso and Juan Miguel Arrazola",
    year = "2024",
    month = "9",
    journal = "PennyLane Demos",
    publisher = "Xanadu",
    howpublished = "\url{https://pennylane.ai/demos/tutorial_intro_qrom}",
    note = "Date Accessed: 2026-08-18"
}

\end{document}